\documentclass[12pt,a4paper]{article}

\usepackage{amsmath,amssymb,latexsym,amsfonts}
\usepackage{amscd}

\newtheorem{theorem}{Theorem}[section]

\newtheorem{proposition}[theorem]{Proposition}

\newtheorem{corollary}[theorem]{Corollary}

\newtheorem{example}[theorem]{Example}
\newcommand{\pr}{ \noindent{\bf Proof:}\quad }
\newcommand{\qed}{$\Box$}
\newenvironment{proof}{\noindent {\pr}\ }{\qed}

\newcommand{\Z}{{\mathbb{Z}}}

\newcommand{\F}{{\mathbb{F}}}
\newcommand{\C}{{\cal C}}
\newcommand{\G}{{\cal G}}
\newcommand{\SG}{{\cal S}}
\renewcommand{\H}{{\cal{H}}}
\newcommand{\zero}{{\mathbf{0}}}

\newcommand{\add}{\Z_2\Z_4}
\newcommand{\zz}{\Z_2^\alpha\times\Z_4^\beta}
\newcommand{\type}{{(\alpha,\beta;\gamma,\delta;\kappa)}}
\newcommand{\pes}{{{\rm wt}}}
\newcommand{\pesL}{{\omega t_L}}

\begin{document}

\title{Permutation decoding of $\add$-linear codes\thanks{This work was partially supported by the
Spanish MICINN under Grants MTM2009-08435 and TIN2010-17358, and
by the Catalan AGAUR under Grant 2009SGR1224.}}

\author{Jos\'{e} Joaqu\'{i}n Bernal,\thanks{
Dept. Matem\'{a}ticas, Universidad de Murcia, Spain.}
Joaquim Borges, Cristina Fern\'{a}ndez-C\'{o}rdoba,\and Merc\`{e}
Villanueva\thanks{
Dept. Information and Communications Engineering, Universitat Aut\`{o}noma de
Barcelona, Spain.}}

\maketitle

\begin{abstract}
An alternative permutation decoding method is described which can be used for
any binary systematic encoding scheme, regardless whether the code is linear or not.
Thus, the method can be applied to some important codes such as $\add$-linear
codes, which are binary and, in general, nonlinear codes in the usual sense. For
this, it is proved that these codes allow a systematic encoding scheme.
As a particular example, this permutation decoding method is applied to some Hadamard $\add$-linear codes.
\end{abstract}

\section{Introduction}\label{introduction}

We denote by $\F^n$ the set of all binary vectors of length $n$ and by $\pes(v)$
the {\em (Hamming) weight} of any vector $v\in\F^n$, that is, the number of its
nonzero coordinates. The {\em (Hamming) distance} between two vectors
$u,v\in\F^n$ is defined as $d(u,v)=\pes(u+v)$. Given a binary code of length
$n$, $C\subseteq \F^n$, we denote by $d_C$ its {\em minimum distance}, that is,
the minimum distance between any pair of different codewords in $C$. We say that $C$ is a
$t$-{\em error-correcting} code, where $t=\lfloor (d_C-1)/2 \rfloor$.

For a vector $v\in\F^n$ and a set $I\subseteq \{1,\dots,n\}$, $|I|=k$, we define $v_I\in\F^k$ as the vector $v$ restricted to the $I$ coordinates. For example,
if $I=\{1,\ldots,k\}$ and $v=(v_1,\ldots,v_n)$, then $v_I=(v_1,\ldots,v_k)$. If
$C$ is a binary code of length $n$, then $C_I=\{v_I: v\in C\}$.

If $C$ has size $|C|=2^k$, then $C$ is a {\em systematic} code if there is a set
$I \subseteq \{1,\ldots,n\}$ of $k$ coordinate positions such that
$|C_I|=2^k$. In other words, $C_I$ is $\F^k$.
Such a set $I$ is also referred to as a set of {\em systematic coordinates} or an
{\em information set}. Given a systematic code of size $\left|C\right|=2^k$ with information set $I$, a {\em systematic encoding for $I$}
is a one-to-one map $f:\;\F^k\;\rightarrow\;\F^n$, such that for any
information vector $a\in\F^k$, the corresponding codeword $f(a)$ satisfies that $f(a)_I=a$.

Let us consider the group of permutations on $n$ symbols, $\SG_n$, acting on $\F^n$ by permuting the coordinates of each vector. That is, for every $v=(v_1,\dots,v_n)\in \F^n$ and $\pi \in \SG_n$, $\pi(v_1,\dots,v_n)=(v_{\pi^{-1}(1)},\dots,v_{\pi^{-1}(n)})$. Then, for any binary code $C$, we denote by $\mbox{PAut}(C)$ its {\em permutation automorphism group}, i.e., $\mbox{PAut}(C)=\{\pi \in \SG_n: \pi(C)=C\}$. Moreover, a binary code $C'$ is said to be {\em permutation equivalent} to $C$ if there exists $\pi\in \SG_n$ such that $\pi(C)=C'$.

Not every binary code of size $2^k$ is systematic, but every binary linear code is systematic. Indeed, if $C\subseteq \F^n$ is a binary linear code of dimension $k$, it is permutation equivalent to a code with generator and parity check matrices:
\begin{equation}\label{eq:matrius}
G=\left(
    \begin{array}{cc}
      Id & A \\
    \end{array}
  \right)
  \;,\;\;\;\;
H=\left(
    \begin{array}{cc}
      A^T & Id \\
    \end{array}
  \right),
\end{equation}
  where $Id$ denotes the identity matrix and $A^T$ is the transpose of $A$. In general, for any information set $I$, we say that a generator (resp. parity check) matrix is in {\em standard form} if the columns in  the positions inside (resp. outside of) $I$ are the columns of $Id$. Then the map $f:\F^k\rightarrow \F^n$ given by
  \begin{equation}\label{linearencoding}
   f(v)=v\ G,
  \end{equation}
  for any $v\in \F^k$, is clearly a systematic encoding.
\medskip

Permutation decoding was introduced in \cite{Prange} and
\cite{McW}. A description of the standard method for linear codes can be found
in \cite[p.513]{MS}. Given a $t$-error-correcting linear code $C\subseteq\F^n$ with fixed information set $I$, we consider $y=x+e$ the received vector, where $x\in C$ and $e$ is the error vector. We assume that $y$ has less than $t+1$ errors, that is, $\pes(e)\leq t$. The idea of permutation decoding is to use the elements of $\mbox{PAut}(C)$ in order to move the non-zero coordinates of $e$ out of $I$. So, on the one hand the method is based on the existence of some special subsets $S\subseteq \mbox{PAut}(C)$, called PD-sets, verifying that for any vector $e\in\F^n$, with $\pes(e)\leq t$, there is an element $\pi\in S$ such that
$\pes(\pi(e)_I)=0$.
On the other hand, the main tool of this decoding algorithm is the following
theorem which gives us a necessary and sufficient condition for a received
vector $y\in \F^n$ having its systematic coordinates correct.

\begin{theorem}[\cite{MS}]\label{teosindrome}
Let $C$ be a $t$-error-correcting linear code with information set $I$ and parity check matrix $H$ in standard form. Let $y=x+e$, where $x\in C$ and $e$ verifies that ${\rm wt}(e)\leq t$. Then
\begin{equation}\label{eq:pes}
{\rm wt}(Hy^T)={\rm wt}(He^T)\leq t\;\;\Longleftrightarrow\;\;{\rm wt}(e_I)=0.
\end{equation}
\end{theorem}

Then, let $C\subseteq \F^n$ be a linear code with information set $I$ and parity check matrix $H$ in standard form. Assume that we have found a PD-set for the information set $I$,  $S\subseteq \mbox{PAut}(C)$, and denote by $y=x+e$ the received vector, where $x \in C$ and $e$ is the error vector. Then the permutation decoding algorithm works as follows:
\begin{enumerate}
\item If $\pes(H y^T)\leq t$, then the systematic coordinates of $y$
are correct and we can recover $x$ from (\ref{linearencoding}).
\item Else, we search $\pi\in S$ such that
$\pes(H \pi(y)^T)\leq t$. If there is
no such $\pi$, we conclude that more than $t$ errors have occurred.
\item If we have successfully found $\pi$, we take $x'$ the unique codeword such that $x'_I=\pi(y)_I$. Then the decoded vector is
$\pi^{-1}\left(x'\right)$.
\end{enumerate}

In this paper, we show that $\Z_2\Z_4$-linear codes are also systematic.
Moreover, we give a systematic encoding method.
However, for $\add$-linear codes, Theorem~\ref{teosindrome}
holds just in some obvious cases, not in general. Nevertheless, we give an alternative method for permutation
decoding which does not need (\ref{eq:pes}). Such method does not use the
syndrome $Hy^T$ to check whether the systematic coordinates are correct or not.
Therefore, the method
can be used for general $\Z_2\Z_4$-linear codes, of course, assuming that we
know an appropriate PD-set.

\medskip

The paper is organized as follows. In Section \ref{sec:systematic}, we show that
any $\add$-linear code is systematic. Moreover, we give an information set and a
systematic encoding for that information set. We also see under which
conditions the standard permutation decoding method works for $\add$-linear
codes. In Section \ref{sec:method}, we present the alternative permutation decoding
method. Such method does not use the syndrome of a received vector in order to
check whether the systematic coordinates are correct or not. Hence, it is not
important if Theorem~\ref{teosindrome} holds or not.
We use the new method for the particular case of  Hadamard $\add$-linear codes. These are, in general, nonlinear codes in the binary sense, but they have high error-correcting capability.

\section{Systematic encoding for $\Z_2\Z_4$-linear codes}\label{sec:systematic}

For every pair $n_1$, $n_2$ of nonnegative integers, define the componentwise Gray map $\Phi: \;\Z_2^{n_1}\times\Z_4^{n_2} \;\longrightarrow\; \F^{n_1+2n_2}$ as
$$\begin{array}{lc}
\Phi(x,y) = (x,\phi(y_1),\ldots,\phi(y_{n_2}))\\
 \hspace{1truecm}\forall x\in\Z_2^{n_1},\;\forall y=(y_1,\ldots,y_{n_2})\in \Z_4^{n_2};
\end{array}$$
where $\phi:\Z_4 \longrightarrow \Z_2^{2}$ is the usual Gray map,
that is, $$\phi(0)=(0,0),\ \phi(1)=(0,1), \ \phi(2)=(1,1), \
\phi(3)=(1,0).$$

The \textit{Lee weights} over the elements in $\Z_4$ are defined as $\pesL(0)=0, \pesL(1)=\pesL(3)=1, \pesL(2)=2$. Then the \textit{Lee weight} of a vector $u=(u_1,\dots,u_{n_1+n_2})\in \Z_2^{n_1}\times\Z_4^{n_2}$ is defined as $\pesL(u)=\pes(u_1,\dots,u_{n_1})+\sum_{i=1}^{n_2}\pesL(u_{n_1+i})$. Finally the \textit{Lee distance} between two vectors $u, v\in \Z_2^{n_1}\times\Z_4^{n_2}$ is $d_L(u,v)=\pesL(u-v)$. The Gray map is an isometry which transforms Lee distances to Hamming distances.

Let $C$ be a $\add$-linear code of type $\type$, length $n=\alpha+2\beta$ and size $|C|=2^k=2^{\gamma+2\delta}$ \cite{Z2Z4Lineals}. As usual, denote by $\C$ the corresponding $\add$-additive code, i.e. $\C=\Phi^{-1}(C)$.
If $\C$ is a $\add$-additive code of type $\type$, it is permutation equivalent to a $\add$-additive code
with generator matrix as follows \cite{Z2Z4Lineals}:
\begin{equation}\label{eq:standard}
\G = \left ( \begin{array}{cc|ccc}
I_{\kappa} & T_b & 2T_2 & \zero & \zero\\
\zero & \zero & 2T_1 & 2I_{\gamma-\kappa} & \zero\\
\hline \zero & S_b & S_q & R & I_{\delta} \end{array} \right ),
\end{equation}
\noindent where $T_b, S_b$ are matrices over
$\Z_2$;  $T_1, T_2, R$ are matrices over $\Z_4$ with all entries in $\{0,1\}\subset \Z_4$; and $S_q$ is a matrix
over $\Z_4$. We say that $\G$ is a matrix in standard form for a $\add$-additive code.

Given two vectors $u,v\in \zz$, the inner product is defined as in \cite{Z2Z4Lineals}:
$$ \langle u,v \rangle=2(\sum_{i=1}^{\alpha} u_iv_i)+\sum_{j=\alpha+1}^{\alpha+\beta}
u_jv_j \ \in \ \Z_4,$$ where the computations are made taking the zeros and ones in the $\alpha$ binary coordinates as quaternary zeros and ones, respectively. The additive dual code of $\C$, denoted by $\mathcal{C}^\perp$, is then defined in the standard way:
$$
\mathcal{C}^\perp=\{y\in\zz\mid \langle x,y \rangle=0, \mbox{ for all } x\in {\cal C}\}.
$$
It is also shown in  \cite{Z2Z4Lineals} that if $\C$ has a generator matrix in standard form
(\ref{eq:standard}), then $\C^\perp$ can be generated by the matrix:
\begin{equation}\label{eq:H}
\H=\left(
\begin{array}{cc|ccc}
T_b^{t} & I_{\alpha-\kappa} & \zero &  \zero & 2S_b^{t}\\
\zero & \zero & \zero & 2I_{\gamma-\kappa}& 2R^{t}\\
\hline T_2^{t} &\zero & I_{\beta+\kappa-\gamma-\delta} &  T_1^{t} & -\big(
S_q+RT_1\big)^{t} \end{array} \right),
\end{equation}
which also represents a parity check matrix for $\C$.
Moreover, $\mathcal{C}^\perp$ is a
$\add$-additive code of type
$(\alpha,\beta;\bar{\gamma},\bar{\delta};\bar{\kappa})$, where
\begin{equation}
\label{dualtype}
\begin{array}{l}
\bar{\gamma} = \alpha + \gamma - 2\kappa,\\
\bar{\delta} =\beta - \gamma - \delta + \kappa,\\
\bar{\kappa}=\alpha-\kappa.
\end{array}
\end{equation}

\medskip
There are some cases where the systematic encoding of $C$ is clear.
First, when $C$ is linear, we can apply simply the standard systematic encoding
for linear codes by considering the generator matrix $G$ of $C$ as in
(\ref{eq:matrius}). It is clear that $C$ is linear, for
example, when $\beta=0$ and also when $\delta=0$. In general, if
$\G$ is the generator matrix of a $\C$ as in (\ref{eq:standard}) and
$\{u_i\}_{i=1}^{\gamma}$ and $\{v_j\}_{j=1}^{\delta}$ are the row vectors of
order two and order four in $\G$, respectively, then $C$ is a binary linear
code if and only if $2v_j*v_k \in \C$, for all $j, k$ satisfying $1\leq j < k
\leq \delta$, where $*$ is the component-wise product (see \cite{FPV}).

The second case is when $\gamma=\kappa$. If we consider the generator and the
parity check
matrix of a code $\C$ with $\gamma=\kappa$, then we obtain
\begin{equation}\label{eq:free}
\begin{array}{cc}
\G = \left ( \begin{array}{cc|cc}
I_{\kappa} & T_b & 2T_2 &  \zero\\
\hline \zero & S_b & S_q &  I_{\delta}
\end{array} \right ),
&
\H=\left(
\begin{array}{cc|cc}
T_b^{t} & I_{\alpha-\kappa} &  \zero & 2S_b^{t}\\
\hline T_2^{t} &\zero & I_{\beta+\kappa-\gamma-\delta} & -S_q^{t} \end{array}
\right).
\end{array}
\end{equation}
It is clear that for any information vector $(u,v)\in \Z_2^\gamma\times\Z_4^\delta$, we
have that $(u,v)\G=(u,z,v)$ for some $z\in \Z_2^{\alpha-\gamma}\times
\Z_4^{\beta-\delta}$ and, therefore, the
set $I=\{1,\dots,\kappa,\alpha+\beta-\delta+1,\dots,\alpha+\beta\}$ is a set of
systematic coordinates. Hence, we have the following systematic
encoding:
\begin{equation} \label{eq:encoding0}
f(a)=\Phi\left(\Phi^{-1}(a)\G\right),\;\;\;\;\forall a\in\F^k.
\end{equation}

Even though in those cases the systematic encoding is clear, we can not use the
same method to $\add$-linear codes in general. We are going to define a method
that use the $\add$-linearity of the code and can be used for all
values of $\alpha,\beta,\gamma,\delta$ and $\kappa$.

\medskip
Let us consider $\C$ a $\add$-additive code of type $\type$ with a generator
matrix in standard form (\ref{eq:standard}), $C=\Phi(\C)$. For each quaternary
coordinate $\alpha+i$, with $i\in\{1,\dots,\beta\}$, we denote by
$\varphi_1(\alpha+i)$ and $\varphi_2(\alpha+i)$ the corresponding pair of binary
coordinates in $\{1,\dots,\alpha+2\beta\}$, that is,
$\varphi_1(\alpha+i)=\alpha+2i-1$ and $\varphi_2(\alpha+i)=\alpha+2i$.
We define the following sets of coordinate positions in $\{1,\ldots,\alpha+2\beta\}$:
\begin{itemize}
\item $J_1=\{1,\ldots,\kappa\}$, $|J_1|= \kappa$.
\item $J_2=\{j_1,\ldots,j_{\gamma-\kappa}\}$, where $j_i=
\varphi_1(\alpha+\beta+\kappa-\gamma-\delta+i)$, $|J_2|=\gamma-\kappa$.
\item
$J_3=\{\varphi_1(\alpha+\beta-\delta+1),\varphi_2(\alpha+\beta-\delta+1),\ldots,
\varphi_1(\alpha+\beta),\varphi_2(\alpha+\beta)\}$, $|J_3|=2\delta$.
\end{itemize}

We are going to show that $J=J_1\cup J_2 \cup J_3$ is a set of systematic
coordinates for the $\add$-linear code $C$. We shall refer to $J$ as the
\textit{standard information set} or \textit{standard set of systematic
coordinates}.

\medskip
Given an information vector
$a=(a_1,\ldots,a_{\gamma+2\delta})\in\F^{\gamma+2\delta}$, we consider the
representation $a=(b,c,d)$, where $b=(a_1,\ldots,a_\kappa)$,
$c=(c_1,\ldots,c_{\gamma-\kappa})=(a_{\kappa+1},\ldots,a_\gamma)$ and
$d=(a_{\gamma+1},\ldots,a_{\gamma+2\delta})$. Note that
$\Phi^{-1}(a)=(b,c,\Phi^{-1}(d))\in\Z_2^\gamma\times\Z_4^\delta$. Consider the
codeword $x=\Phi^{-1}(a)\G\in \Z_2^\alpha \times \Z_4^\beta$. For each $i\in \{1,\ldots,\gamma-\kappa\}$, define
\begin{equation}\label{eta}
\eta_i=\left\{\begin{array}{cl}
         0 & \mbox{ if } x_{j_i}=c_i, \\
         1 & \mbox{ otherwise,}
       \end{array}\right.
\end{equation}
where, following the notation given above, $J_2=\{j_1,\dots,j_{\gamma-\kappa}\}$.
Let $\eta=(\eta_1,\ldots,\eta_{\gamma-\kappa})$. Then, we consider the bijection
$\sigma:\;\Z_2^\gamma\times\Z_4^\delta\;\longrightarrow\;\Z_2^\gamma
\times\Z_4^\delta$ given by
$\sigma(\Phi^{-1}(a))=\sigma(b,c,\Phi^{-1}(d))=(b,c+\eta,\Phi^{-1}(d))$. It is
straightforward that the codeword $\sigma(\Phi^{-1}(a)) \G$ verifies that
$$\left(\Phi\left(\sigma(\Phi^{-1}(a))\G\right)\right)_{J}=(b,c,d)=a.$$

Since $|J|=\kappa+\gamma-\kappa+2\delta=\gamma+2\delta$, we conclude that $J$ is a set of systematic coordinates. Therefore we have proved the following theorem.

\begin{theorem}\label{systematic}
If $C$ is a $\add$-linear code of type $\type$, then $C$ is a systematic code. Moreover, if we assume that the generator matrix of $\C=\Phi^{-1}(C)$ is in standard form (\ref{eq:standard}), then $J=J_1\cup J_2\cup J_3$ is a set of systematic coordinates for $C$.
\end{theorem}

Note that in the case $\gamma=\kappa$ we have $a=(b,d)$, so $\eta$ is the
all-zero vector and hence $\sigma$ is the identity map.
Therefore, as a result, for $\gamma=\kappa$ we obtain the systematic encoding function
given in (\ref{eq:encoding0}).

\begin{corollary}
Let $C$ be a $\add$-linear code of length $n$, size $|C|=2^k$ and such that $\C=\Phi^{-1}(C)$ has generator matrix in standard form (\ref{eq:standard}). Then, the function
$f:\;\F^k\;\longrightarrow\;\F^n$ defined as
\begin{equation} \label{eq:encoding}
f(a)=\Phi\left(\sigma\left(\Phi^{-1}(a)\right)\G\right),\;\;\;\;\forall a\in\F^k
\end{equation}
is a systematic encoding for $C$ and the information set $J$.
\end{corollary}

The following example shows that the set of systematic coordinates is not unique, in general.

\begin{example}
Consider the $\add$-additive code $\C$ generated by
$$
\G=\left(
    \begin{array}{cc|ccc}
      1 & 1 & 2 & 0 & 0 \\
      0 & 0 & 2 & 2 & 0 \\
      \hline
      0 & 1 & 1 & 1 & 1 \\
    \end{array}
  \right).
$$
Let $C=\Phi(\C)$ be the corresponding $\add$-linear code in $\F^8$. A set of systematic coordinates is $\{2,4,6,8\}$. However, the standard set of systematic coordinates would be $\{1,5,7,8\}$.
\end{example}

Note that this encoding method requires, in some cases, two products by the generator matrix. However, this is not a meaningful change of complexity order.

\section{An alternative permutation decoding algorithm}\label{sec:method}

In this section we are going to see that the usual permutation decoding algorithm can be applied to $\add$-linear codes just in a few cases. This is because, even if we find a PD-set, Theorem~\ref{teosindrome} can not be used in general. We shall present an alternative permutation decoding algorithm where Theorem~\ref{newteosindrome} replaces Theorem~\ref{teosindrome}.

Let $C$ be a $t$-error correcting $\add$-linear code with information set $J$. Let $\C=\Phi^{-1}(C)$ be the corresponding $\add$-additive code of type $\type$. On the one hand, we have seen that if $C$ is linear then the usual systematic encoding can be applied considering the matrices as in (\ref{eq:matrius}). So Theorem~\ref{teosindrome} works.

On the other hand, if $\gamma=\kappa$ then we have seen that we can assume that
$\C$ has a parity check matrix $\H$ containing the identity matrix (see
(\ref{eq:free})). Then, denote the received vector $y=x+e\in
\F^{\alpha+2\beta}$, where $x\in C$ and $e$ is the error vector. It is easy to
see that we may adapt Theorem~\ref{teosindrome} to this context, that is, we
have that, under the condition $\pes(e)\leq t$,
\begin{equation}\label{teosindromez2z4}
\pesL(\H \Phi^{-1}(y)^T)=\pesL(\H \Phi^{-1}(x)^T)\leq
t\;\;\Longleftrightarrow\;\;\pes(e_J)=0.
\end{equation}
where, recall that $\pesL()$ represents the Lee weight. The following result shows that in the nonlinear case, (\ref{teosindromez2z4}) holds if and only if $\gamma=\kappa$.

\begin{proposition}
Let $\C$ be a $t$-error-correcting $\add$-additive code of length $n$, type
$\type$ and parity check matrix $\H$, such that $C=\Phi(\C)$ is a binary
nonlinear code. Then, $C$ satisfies (\ref{teosindromez2z4}) if and only
if $\gamma=\kappa$.
\end{proposition}
\medskip

\begin{proof}

 The case $\gamma=\kappa$ have been discussed above. So assume $\gamma>\kappa$.

 Denote by $e_i$ the binary vector of length $n$ which has weight one and its
nonzero coordinate is at position $i$ ($1\leq i \leq n$). Define the three binary coordinate sets:
\begin{itemize}
 \item $L_1=\{\kappa+1,\cdots,\alpha\}$,
 \item $L_2=\{\varphi_1(\alpha+1),
\varphi_2(\alpha+1), \ldots,
\varphi_1(\alpha+\beta+\kappa-\gamma-\delta),
\varphi_2(\alpha+\beta+\kappa-\gamma-\delta)\}$,
\item $L_3=\{j_1,\ldots,j_{\gamma-\kappa}\}$, where $j_i=
\varphi_2(\alpha+\beta+\kappa-\gamma-\delta+i)$.
\end{itemize}

We have that $L=L_1\cup L_2\cup L_3=\{1,\dots,n\}\setminus J$. Consider an error
vector $e\in\F^n$ such that $\pes(e)=t$, $\pes(e_J)=0$ and $\pes(e_{L_3})\not=0$.
By the definition of $\H$, (\ref{eq:H}), it is easy to check that
for $k_1,\dots,k_r\in L_3$ we have that $\pesL(\H \Phi^{-1}(e_{k_1}+\dots+
e_{k_r})^T)\geq
2r$, and $\pesL(\H \Phi^{-1}(e)^T)=\pesL(\H \Phi^{-1}(\varepsilon_1)^T)+\pesL(\H
\Phi^{-1}(\varepsilon_2)^T)$, where $\varepsilon_1=(\varepsilon_1^1,\dots,\varepsilon_1^n)\in\F^n$ is given by $(\varepsilon_1)_{L_1}=e_{L_1}$, $\varepsilon_1^i=0$ if $i\notin L_1$, and $\varepsilon_2=(\varepsilon_2^1,\dots,\varepsilon_2^n)\in\F^n$ is given by $(\varepsilon_2)_{L\setminus L_1}=e_{L\setminus L_1}$, $\varepsilon_2^i=0$ if $i\in L_1$. Hence, if
$\pes(e_{L_2})=0$, we obtain $\pesL(\H
\Phi^{-1}(e)^T)> t$. In the case $\pes(e_{L_2})\not=0$, there exist $i\in
L_2$ such that  $\pes(e+e_i)=t-1$ and $\pesL(\H \Phi^{-1}(e+e_i)^T)>t$. In both
cases we have that $C$ does not satisfy (\ref{teosindromez2z4}).
\end{proof}
\medskip


\begin{theorem}\label{newteosindrome}
Let $C$ be a binary systematic $t$-error-correcting code of length $n$.
Let $I$ be a set of systematic coordinates and
let $f$ be a systematic encoding for $I$. Suppose that
$y=x+e$ is a received vector, where $x\in C$ and ${\rm wt}(e)\leq t$. Then, the
systematic coordinates of $y$ are correct, i.e. $y_I=x_I$, if and only if
${\rm wt}\left(y+f(y_I)\right)\leq t$.
\end{theorem}

\begin{proof}
If $\pes\left(y+f(y_I)\right)\leq t$, then $f(y_I)$ is the closest codeword to
$y$, that is, $f(y_I)=x$. Hence the systematic coordinates are the same $y_I=x_I$.

If $x_I=y_I$, then $\pes\left(y+f(y_I)\right)=\pes (y+x)=\pes(e)\leq t$.
\end{proof}

\medskip

Now, let us consider $\C$ a $\add$-linear code with information set $I$. Assume that $S\subseteq \mbox{PAut}(C)$ is a PD-set for $I$ and $y$ is a received vector. As an alternative method with respect to the algorithm described in Section~\ref{introduction} we can use the following decoding process:

\begin{enumerate}
\item If $\pes\left(y+f(y_I)\right)\leq t$, then $x=f(y_I)$ is the decoded
vector and $y_I$ is the information vector.
\item Else, we search $\pi\in S$ such that
$\pes\left(\pi(y)+f\left(\pi(y)_I\right)\right)\leq t$. If there is
no such $\pi$, we conclude that more than $t$ errors have occurred.
\item If we have successfully found $\pi$, then the decoded vector is
$$x=\pi^{-1}\left(f\left(\pi(y)_I\right)\right).$$
\end{enumerate}

Note that $\pes\left(\pi(y)+f\left(\pi(y)_I\right)\right)\leq t$
implies that $f\left(\pi(y)_I\right)$ is the closest codeword to
$\pi(y)$. Therefore, the closest codeword to $y$ is
$\pi^{-1}\left(f\left(\pi(y)_I\right)\right)$.

\begin{example}

Consider the $\add$-additive code $\C$ with generator and parity check matrices:
$$ \G=\left (\begin{array}{cccc}
\hline
3&2&1&0\\
2&3&0&1 \end{array} \right ), \quad \quad  \H=\left (\begin{array}{cccc}
\hline
1&0&1&2\\
0&1&2&1 \end{array} \right ).$$

The corresponding $\add$-linear code $C=\Phi(\C)$ is a 1-error-correcting code of type $(0,4;0,2;0)$ (i.e., $C$ is a $\Z_4$-linear code). In fact, $C$ is a Hadamard $\Z_4$-linear code \cite{Kro01}. Let $\vartheta=(1,3,5,7)(2,4,6,8)$. It is straightforward to check that $\vartheta\in \mbox{PAut}(C)$ (note that $\C$ is a quaternary cyclic code) \cite{PePuVi12}. Moreover, $S=\{id,\vartheta,\vartheta^2\}$ is a PD-set for the standard information set $I=\{5,6,7,8\}$.
Since $\gamma=\kappa$, we can use the systematic encoding $f$ defined in
(\ref{eq:encoding0}).
%

For example, let $a=(0,1,0,1)\in \F^4$ be an information vector. Then
$$
x=f(a)=\Phi\left(\Phi^{-1}(a)\G\right)=\Phi\left((1,1)\G\right)=\Phi(1,1,1,1)=(0,1,0,1,0,1,0,1).
$$
Suppose now that the received vector is $y=x+e$, where $e=(0,0,0,0,0,0,0,1)$. The syndrome of $y$ is
$$
\Phi\left(\H\Phi^{-1}(y)^T\right)=\Phi\left(\H(1,1,1,0)^T\right)=\Phi((2,3)^T)=(1,1,1,0)^T,
$$
which has weight $3> t=1$. However, considering the vector $z=\vartheta(y)=(0,0,0,1,0,1,0,1)$, we have that the syndrome is
$$
\Phi\left(\H\Phi^{-1}(z)^T\right)=\Phi((3,0)^T)=(1,0,0,0)^T,
$$
which has weight $1\leq t=1$. Therefore, the systematic coordinates of $z$ have no errors. Hence, we decode $y$ as
\begin{eqnarray*}
\vartheta^{-1}\left(\Phi(\Phi^{-1}(z_I)\G)\right) &=& \vartheta^{-1}\left(\Phi((1,1)\G\right))=\vartheta^{-1}(\Phi(1,1,1,1))\\
&=& (0,1,0,1,0,1,0,1) =x,
\end{eqnarray*}
and the information vector is $x_I=(0,1,0,1)$.

\end{example}

\begin{example}
Consider the $\add$-additive code $\C$ with generator matrix:
$$ \G=\left (\begin{array}{cccccccc}
2&2&2&0& 0&2&0&0 \\ \hline
3&2&1&2& 3&0&1&0 \\
2&3&0&3& 2&1&0&1 \end{array} \right ).$$

The corresponding $\add$-linear code $C=\Phi(\C)$ is a 3-error-correcting code of type $(0,8;1,2;0)$ (i.e., $C$ is a $\Z_4$-linear code). In fact, $C$ is also a Hadamard $\Z_4$-linear code \cite{Kro01}. We know that $\langle  \vartheta_1, \vartheta_2, \vartheta_3, \vartheta_4 \rangle \subseteq  \mbox{PAut}(C)$ \cite{PePuVi12}, where
\begin{equation}
\begin{split}
\vartheta_1=& (1,5)(2,6)(3,11)(4,12)(9,13)(10,14)(7,15)(8,16),\\
\vartheta_2=& (1,3,5,11)(2,4,6,12)(9,7,13,15)(10,8,14,16),\\
\vartheta_3=& (9,13)(10,14)(7,15)(8,16), \\
\vartheta_4=& (1,9)(2,10)(5,13)(6,14).
\end{split}
\end{equation}
Moreover, it is easy to check using the {\sc Magma} software package \cite{M1} that we can take the elements in the subgroup $S=\langle \vartheta_1, \vartheta_2,
\vartheta_4 \rangle$ as a PD-set for the information set $I=\{11,13,14,15,16\}$.
In this case, we can not use the standard permutation decoding, since $\gamma\not =\kappa$. However, we can still
perform a permutation decoding using the alternative method presented in this section.

For example, let $a=(1,1,1,1,1) \in \F^5$ be an information vector. Using the systematic encoding given by (\ref{eq:encoding}),
the corresponding codeword is
\begin{eqnarray*}
x &=&
f(a)=\Phi\left(\sigma(\Phi^{-1}(a))\G\right)=\Phi\left((1+\eta_1,2,2)\G\right)
\\
  &=& \Phi(2,2,2,2,2,2,2,2)=(1,1,1,1,1,1,1,1,1,1,1,1,1,1,1,1),
\end{eqnarray*}
where $\eta=(\eta_1)=(1)$. Suppose now that the received vector is $y=x+e$,
where $e=(0,0,0,0,0,0,0,0,0,0,0,0,1,0,1,1)$.
By considering the standard information set, the information coordinates of $y$ are $y_I=(1,0,1,0,0)$ and
$$
f(y_I)=\Phi(\sigma(\Phi^{-1}(y_I))\G )=(0,1,0,0,1,0,1,1,1,0,1,1,0,1,0,0),
$$ so $\pes(y+f(y_I))=5>t=3$. However, considering the vector $z=\vartheta_1(y)=(1, 1, 1, 1, 1, 1, 0, 0, 0, 1, 1, 1, 1, 1, 1, 1)$,
we have that  $z_I=(1,1,1,1,1)$ and
$$f(z_I)= \Phi(\sigma(\Phi^{-1}(z_I))\G )=(1,1,1,1,1,1,1,1,1,1,1,1,1,1,1,1),$$
so $\pes(z +f(z_I))=3\leq t=3$.
Therefore, the systematic coordinates of $z$ have no errors. Hence, we decode $y$ as
$\vartheta_1^{-1}(f(z_I))=x$ and the information vector is $x_I=(1,1,1,1,1)$.
\end{example}


\section*{Acknowledgements}
The authors thanks Prof. J. Rif\`{a} for valuable discussions in an earlier version of this paper.


\begin{thebibliography}{1}
%
\bibitem{Key} E.F. Assmus and J.D. Key, {\em Designs and
their codes}, Cambridge University Press, Great Britain, 1992.

\bibitem{Z2Z4Lineals} J. Borges, C. Fern\'{a}ndez-C\'{o}rdoba, J. Pujol, J. Rif\`{a} and M. Villanueva,
``$\add$-linear codes: generator matrices and duality," {\em Designs, Codes and Cryptography},
vol. 54, pp. 167-179, 2010.

\bibitem{BoPhRi03} J. Borges, K.T. Phelps and J. Rif\`{a}, ``The rank and kernel of
extended 1-perfect $\Z_4$-linear and additive non-$\Z_4$-linear codes,"
\emph{IEEE Trans. on Information Theory}, vol. 49(8), pp. 2028-2034, 2003.

\bibitem{M1} J.J. Cannon and W. Bosma (Eds.) {\it Handbook of {\sc Magma} Functions},
Edition 2.13, 4350 pages, 2006.

\bibitem{FPV} C. Fern\'andez-C\'ordoba, J. Pujol and M. Villanueva,
``$\Z_2\Z_4$-linear codes: rank and kernel,'' {\em Designs Codes and Cryptography},
vol. 56, pp. 43-59, 2010.

\bibitem{Kro01} D.S. Krotov, ``$\Z_4$-linear Hadamard and
  extended perfect codes," \emph{Electron. Notes in Discr. Math.},
  vol. 6, pp. 107-112, 2001.

\bibitem{Kro12} D.S. Krotov, ``On the automorphism groups of the additive 1-perfect binary codes,"
Proceedings of the 3rd International Castle Meeting on Coding Theory and Applications, Cardona, Spain,
pp. 171-176, 2011.

\bibitem{McW} F.J. MacWilliams, ``Permutation decoding of systematic codes," {\em Bell Syst. Tech. J.}, vol. 43, pp. 485-505, 1964.

\bibitem{MS} F.J. MacWilliams and N.J.A. Sloane, \emph{The Theory of
  Error-Correcting Codes}, North-Holland Publishing Company, 1977.

\bibitem{PePuVi12} J. Pernas, J. Pujol and M. Villanueva, Characterization of the
automorphism group of quaternary linear Hadamard codes, \emph{Designs, Codes and
Cryptography} (2012), DOI 10.1007/s10623-012-9678-2.

\bibitem{Prange} E. Prange. ``The use of information sets in decoding cyclic codes," {\em IEEE Trans. Info. Theory}, vol. 8, no. 5, pp. S5-S9, 1962.

%
\end{thebibliography}
\end{document}